\documentclass[12pt]{alt2021} % Anonymized submission
\usepackage[ruled]{algorithm2e}
\usepackage{amsfonts}
\usepackage{multirow}
\usepackage{diagbox}
\usepackage{slashbox}
\usepackage[table]{xcolor}
\usepackage{nicefrac}

\newtheorem{thm}{Theorem}
\newtheorem{lem}[thm]{Lemma}

\DeclareMathOperator*{\argmin}{arg\,min}

\newtheorem{defn}[thm]{Definition}
\def\D{{\mathcal D}}

\def\R{{\mathbb R}}
\def\X{{\mathcal X}}
\def\E{{\mathbb E}}
\def\Pr{{\mathbb P}}
\def\diam{\text{diam}}
\def\asp{\text{asp}}
\def\c{\mathcal L}
\def\seq{\text{OC}}
\def\C{\mathcal C}

\def\v{v}

\def\o{opt}

\newlength\algowd

\def\mes{OC}
\def\C{\mathcal C}
\newcommand{\poly}{\mathrm{poly}}
\newcommand{\polylog}{\mathrm{polylog}}

\title[No-substitution $k$-means Clustering with Adversarial Order]{No-substitution $k$-means Clustering with Adversarial Order}
\usepackage{times}
\altauthor{%
 \Name{Robi Bhattacharjee} \Email{rcbhatta@eng.ucsd.edu}\\
 \addr UCSD
 \AND
 \Name{Michal Moshkovitz} \Email{mmoshkovitz@eng.ucsd.edu}\\
 \addr UCSD
}

\begin{document}

\maketitle

\begin{abstract}%
 We investigate $k$-means clustering in the online no-substitution setting when the input arrives in \emph{arbitrary} order. In this setting, points arrive one after another, and the algorithm is required to instantly decide whether to take the current point as a center before observing the next point. Decisions are irrevocable. The goal is to minimize both the number of centers and the $k$-means cost. Previous works in this setting assume that the input's order is random, or that the input's aspect ratio is bounded. It is known that if the order is arbitrary and there is no assumption on the input, then any algorithm must take all points as centers. Moreover, assuming a bounded aspect ratio is too restrictive --- it does not include natural input generated from mixture models.

 We introduce a new complexity measure that quantifies the difficulty of clustering a dataset arriving in arbitrary order. We design a new random algorithm and prove that if applied on data with complexity $d$, the algorithm takes $O(d\log(n) k\log(k))$ centers and is an $O(k^3)$-approximation. We also prove that if the data is sampled from a ``natural" distribution, such as a mixture of $k$ Gaussians, then the new complexity measure is equal to $O(k^2\log(n))$. This implies that for data generated from those distributions, our new algorithm takes only $\poly(k\log(n))$ centers and is a $\poly(k)$-approximation.  In terms of negative results, we prove that the number of centers needed to achieve an $\alpha$-approximation is at least $\Omega\left(\frac{d}{k\log(n\alpha)}\right)$.
\end{abstract}


% Acknowledgments---Will not appear in anonymized version
\section{Introduction}
Clustering is a fundamental task in unsupervised learning with many diverse applications such as health \cite{zheng2014breast}, fraud-detection \cite{sabau2012survey}, and recommendation systems \cite{logan2004music} among others. The goal of $k$-means clustering is to find $k$ centers that minimize the $k$-means cost of a given set of points. The cost is the sum of squared distances between a point and its closest center. This problem is NP-hard \cite{aloise2009np,dasgupta2008hardness}, and, consequently, approximated algorithms are used \cite{arthur2006k,aggarwal2009adaptive, kanungo2004local}.  An algorithm is an $\alpha$-approximation if the $k$-means cost of its output is at most $\alpha$ times the optimal one.

In this paper, we focus on the online no-substitution setting where the points arrive one after another, and decisions are made instantly, before observing the next point. In this online  setting, the algorithm decides whether to take a point as center immediately upon its arrival. Decisions cannot be changed; once a point is considered a center, it remains center forever. Conversely, points that are not centers cannot become centers after the next point is received. This set-up is summarized in Algorithm~\ref{alg:intro_online_setting}.
\begin{algorithm}[ht]
\caption{Online no-substitution setting}
\label{alg:intro_online_setting}
\SetAlgoLined
$C \leftarrow \{\}$ \tcp*{\parbox[t]{3.3in}{set of centers}}
\For{$t \in \{1, 2, 3, \dots, n\}$}{
% \algolines{decide whether to add $x_t$ to $C$}{only $x_t$ can be added to $C$ at time $t$ }
decide whether to add $x_t$ to $C$\tcp*{\parbox[t]{3.3in}{\raggedright only $x_t$ can be added to $C$ at time $t$, and no center can be removed from $C$}}
% \atcp{}
%  \atcp{no element can be removed from $C$}
}
return $C$ \;
\end{algorithm}
% \begin{itemize}
% \item $C\rightarrow\emptyset$ \quad\# set of centers 
%     \item For $t = 1, 2,\ldots $ 
%     \begin{itemize}
%         \item decide whether to add $x_t$ to $C$
%         \item \# only $x_t$ can be added to $C$ at time $t$; no element can be removed from $C$
%     \end{itemize}
%     \item At the end check the quality of $C$
% \end{itemize}
In this setting, the goal is two-fold: (i) minimizing the cost of the returned centers, $C$, and (ii) minimizing number of centers, $|C|$.
 
The importance of the online no-substitution setting is motivated by several examples in \cite{hess2020sequential}. Here is one of them: suppose we are running a clinical trial for testing a new drug. Patients come one after another to the clinic, and for each of them we must decide whether or not to administer the drug. Our choices are immutable: once a patient takes the drug, it cannot be untaken, and once a patient leaves the clinic, we cannot decide to test the drug on them. Our overall goal is to administer the drug to a small representative sample of the entire population. In this example, the patients are the points, and the people given the drug are the centers. The number of patients that took the drug is the number of centers, which should be small as it is an experimental drug and thus risky. Assuming an appropriate distance measure between any two people, a low cost $k$-means clustering provides a good representation of the entire population. In this setting, previous works  \cite{hess2020sequential,moshkovitz2019unexpected} present algorithms under the assumption that the order is random. However, this is not always the case. In the motivating example, the elderly patients might tend to arrive earlier than the younger ones. To account for this, in this paper we focus on the case that the order is arbitrary and might be adversarial. 
 
It is known that if $n$ input points arrive in arbitrary order, then $\Omega(n)$ centers are needed, as we demonstrate next. To prove this lower bound for any $\alpha$-approximation algorithm, consider an exponential series of points in $\R$: $(2\alpha)^1,(2\alpha)^2,\ldots.$ This sequence is constructed so that each point is very far away from the points preceding it. At each time step $t$, suppose that the algorithm doesn't select the current point, $(2\alpha)^t$. Because the number of points, $n$, is not known in advance, the algorithm must assume at all times that the series might stop. If this happens, the algorithm cost is roughly $(2\alpha)^t$, because of the cost incurred by the point $(2\alpha)^t$. On the other hand, the cost of the optimal clustering is roughly $(2\alpha)^{t-1}$, which violates the fact that the algorithm is an $\alpha$-approximation. Therefore, the algorithm must select every point in the sequence. Observe that this lower-bound dataset is well-designed and pathological. This leads to the following questions: are there more ``natural" datasets that require $\Omega(n)$ centers, and can we find a property that is shared by all hard-to-cluster datasets? 

In this paper, we define a new measure for datasets that accurately captures the hardness of learning in the no-substitution setting with arbitrary order. We show that the lower-bound pathological example from above is the only reason for necessarily taking a large number of points as centers. We also design a new algorithm that takes a small number of centers whenever the new measure is small.  Additionally, we prove that for data generated from many ``natural" distributions (like Gaussian mixture models), the new measure is small ($O(k^2\log n)$). The reason, intuitively, is that when sampling $n$ points for many distributions, the largest subset with distances exponentially increasing is of length $O(\log n)$. This allows taking only $\polylog(n)$ centers. This number of centers is so small that it is close to the number of centers needed if the order is random \cite{moshkovitz2019unexpected}.

\subsection{Problem setting}
Fix a dataset $X$ with $|X|=n$ points. 
We define a clustering of $X$ both by its clusters and its centers $\C = \{(C^1, c^1), (C^2, c^2), \dots, (C^k, c^k)\},$ where $X=C_1 \sqcup \ldots\sqcup C_k$ is a partition of $X$, and $c^i$ is the center of the cluster $C^i$.
A clustering's $k$-means cost with distance metric\footnote{$d$ should be for all $x,y$  (i) non negative  $d(x,y)\geq 0$ (ii) $d(x,x)=0$  and (iii) symmetric $d(x,y) = d(y,x)$} $d$ is equal to  $$\c(\C)=\sum_{i=1}^\ell\sum_{x\in C^i} d(x,c^i)^2.$$
 The optimal $k$-clustering, $opt_k$, with cost $cost(opt_k)$, is the one the minimizes the cost $$opt_k = \argmin_{\substack{\C=((C^1,c^1),\ldots,(C^k,c^k)) \text{ s.t. }\\  X = C_1 \sqcup \ldots\sqcup C_k}}\c(C).$$
Sometimes we define a clustering only by its centers or only by its clusters, and assume that the missing clusters or centers are the best ones. The best center for a cluster is its mean, and the best clusters for a given set of centers $c^1,\ldots,c^k$ are the closest points to each center $C^i=\{x:i=\argmin_{j\in[k]} d(x,c^j) \}$.

%In the no-substitution setting, the algorithm receives points in $X$ one after another. As each point arrives, the algorithms decide whether to take it as a center or not. The decisions cannot be changed after the next point arrives.  Points that were not taken as centers cannot be considered centers later on, and points that were taken as centers cannot be removed from the set of centers. The clustering algorithm does not know $n$ in advance. 

\begin{center}
\textit{The focus of the paper is the no-substitution setting where points arrive in an \textbf{\emph{arbitrary order}}.} 
\end{center}

An online algorithm should return a clustering that is comparable to $opt_k.$ But, for that to happen, the number of centers must be larger than $k$ \cite{moshkovitz2019unexpected}. Therefore, online algorithms return $\ell\geq k$ centers. We formally define ``comparable to $opt_k$" in the following way. An $\alpha$-approximation algorithm relative to $opt_k$ is one that returns a clustering $\C=((C^1,c^1),\ldots,(C^\ell,c^\ell))$ such that with probability\footnote{The probability is over the randomness of the algorithm; $0.9$ can be replaced by any constant close to $1$, this constant was used in \cite{moshkovitz2019unexpected}.} at least $0.9$, $\c(\C)\leq \alpha\cdot \c(opt_k)$. 
The goal is to find an $\alpha$-approximation algorithm that returns $\ell$ centers and aims to minimize \emph{two} values: both $\alpha$ and $\ell$. There is a trade-off between the two. Here are two extreme examples. If $\ell$ is large and equals $n$, then $\alpha$ is small and even equal to zero.  If $\alpha$ can go to infinity, then $\ell$ can be small and equal to $1$. In this paper, we focus on the case that $k\ll n$ or even $k$ a constant. Thus, $\alpha$ should be independent of $n$ and might be equal to $\poly(k).$

\subsection{Our results}
In this paper, we design an algorithm that can learn in the no-substitution setting, even when points arrive in an \emph{arbitrary} order. We showed that it is not possible for all datasets. Therefore, we introduce a new complexity measure that quantifies the hardness of learning a dataset in this setting. We show that the number of centers taken by the algorithm depends on the new complexity measure, and prove it is small for many datasets. We now formally define the new complexity measure. 

An essential notion is a \emph{diameter} of a set of points.  For a set of points $X$ its diameter, $\diam(X)$, is the distance of the two farthest points, $\diam(X)=\max_{x,x'\in X}d(x,x')$. 
The diameter of a clustering $\C$ is the maximal diameter among all of its clusters, $\diam(C)=\max_{C^i\in\C} \diam(C^i).$ The best-$\ell$-diameter (or $\ell$-diameter for short) of a set of points $X$, $\diam_\ell(X)$, is the diameter of the clustering with $\ell$ clusters and smallest diameter  $$\diam_l(X) = \min_{\substack{\C=((C^1,c^1),\ldots,(C_\ell,c^\ell)) \text{ s.t. }\\  X = C_1 \sqcup \ldots\sqcup C_\ell }}\diam(\C).$$ 
Earlier, we presented an exponential series that required taking $\Omega(n)$ centers. This lower-bound dataset required many centers because there is an order of the dataset where each new point seems to start a new cluster at its arrival. 
Intuitively, a new point $x_t$ starts a new cluster if its distance to the closest point in $X_{t-1}=\{x_1,\ldots, x_{t-1}\}$ is farther than the diameter of the best clustering of $X_{t-1}$ into $k-1$ clustering.
%$diam(C_{t-1})$
More formally, the new complexity measure, Online Center measure, $\seq_k(X)$, is defined as follows. 
\begin{defn}
The Online Center measure of a set $X$, denoted $\seq_k(X)$, is the length of the largest sequence of points in $X$, $x_1,\ldots,x_{m}\in X$,
 such that for every $1 < j \leq m$, $$ \min_{1 \leq i \leq {j-1}} d(x_i, x_j) > 2 \emph{\diam}_{k-1}(\{x_1, x_2, \dots, x_{j-1}\}).$$ 
\end{defn}
The new measure, $\seq_k(X)$, of a dataset $X$ with $|X|=n$  and $k$ clusters is an integer in $[n]$.
The constant $2$ can be replaced by any scalar $\alpha>1$, %by reducing the measure a bit, 
see Lemma \ref{lem_diameter}. 
As a sanity check, note that the lower-bound dataset has the highest $\seq$ possible, $n$, and indeed, the maximal number of centers is required for this dataset. 

In the paper, we show that if the data is sampled from a ``natural" mixture distribution then $\seq_k(X)=O(k^2\log n)$. Many distributions satisfy this: 
\begin{thm}[informal]
With probability of about $1-\frac{1}{n}$, sample $X$ of $n$ points from $k$ mixture of:  Gaussian, uniform, or exponential distributions has $\seq_k(X)=O(k^2\log n).$
\end{thm}

One of the main contributions of this work is designing an algorithm in the no-substitution setting that uses a small number of centers, even if the points' order is adversarial. We prove that Algorithm \ref{alg:1} uses $O(\seq_k(X)\log n)$ centers, when $k$ is constant. 

\begin{thm}[main theorem] \label{thm:main_theorem}
Let $A$ be an $\alpha$-approximation offline algorithm and $X$ a set of points to be clustered. Algorithm \ref{alg:1} gets as an input $X$ in the no-substitution setting in \emph{arbitrary order} and returns a set of centers $S$ with the following properties
\begin{enumerate}
    \item with probability $0.9$, the cost is bounded by $$cost(S)\leq 1358\alpha k^3cost(opt_k)$$
    \item expected number of centers is bounded by  $$\E[|S|] \leq 160\emph{\seq}_k(X)k\log(k)(\log n+1).$$
\end{enumerate} 
\end{thm}
The main theorem gives an upper bound on the expected number of centers taken by Algorithm~\ref{alg:1} and its approximation quality.
 In particular, it takes $O(\seq_k(X)k\log(k) \log(n))$ centers, and it is an $O(\alpha k^3)$-approximation.
This means an exponential increase in the number of centers, without any information on the data.
Note that the algorithm does not need to know what $\seq_k(X)$ is.
Furthermore, if the data is generated from a distribution, the algorithm does not need to know this distribution, and it does not learn this distribution.
 
The idea of the algorithm is the following. For any optimal cluster $C^*_i$, it is well known that if we take each point in $C^*_i$ with probability inversely proportional to $|C^*_i|$, then we get a good center for the entire cluster $C^*_i$. We cannot use this observation as-is because in the online setting, when a point $x_t$ appears, we do not know the cluster $C^*_i$ that it belongs to, and more importantly, we do not know its size, $|C^*_i|.$
The algorithm tries to give an estimate for $|C^*_i|$. 
If this estimate is too small, the algorithm might take too many centers. If the estimate is too large, the algorithm might not take a good center from each optimal cluster. 
To find a reliable estimate, we use, as a subroutine, an approximated clustering algorithm $A$ for the offline setting. After we observe a new point $x_t$ we call $A$ on the points observed so far, including $x_t$ and get clustering $C^{t}$. We use the cluster that $x_t$ is in, $x_t\in C^t_i$, to estimate the probability to take $x_t$. The value $|C^t_i|$ might be too small and unreliable. Therefore, the algorithm merges into $C^t_i$ all clusters that are close to  $C^t_i$. We show that after this merge, the estimation is just right, and we can bound both the approximation and the number of centers taken by the algorithm.

Another contribution of the paper is a proof that $\seq_k(X)$ lower bounds the number of centers taken by any algorithm in the online no-substitution setting. We prove that if $\seq_k(X)$ is large, many centers need to be taken by any algorithm.
\begin{thm}[lower bound] Let $X$ be an arbitrary set of points. There exists an ordering of $X$ such that any $\alpha$-approximation algorithm on $X$, the expected number of centers is $\Omega\left(\frac{\seq_k(X)}{k \log (n\alpha)}\right)$.
\end{thm}

To illustrate our results, in Table~\ref{tab:table-intro}, we summarize some of them for constant $k$. 
Each cell in the table states the number of centers suffice to achieve $\Theta(1)$-approximation. In case the order is random, \cite{moshkovitz2019unexpected} presented an algorithm that takes only $O(\log n)$ centers, no matter what the dataset is, and proved that the  pathological dataset presented earlier provides a matching lower bound. In case the order is arbitrary and there are no assumptions on the dataset, $\Omega(n)$ centers is required, as discussed earlier. In this work, we fill up the gap and show that if $\seq_k(X)=\polylog(n)$, then the number of centers is upper and lower bounded by $\polylog(n),$ which means that, up to a polynomial, we use the same number of centers as the random order case.  

\begin{table}[!h]
\begin{center}
 \begin{tabular}{||c || c c||} 
 \hline
  \backslashbox{\mes}{order}  & random & worst \\  
  \hline\hline
arbitrary  & $\Theta(\log n)$ &$\Omega(n)$   \\
$\polylog (n)$  & $O(\log n)$ &  \cellcolor{blue!8}$\Theta(\polylog(n))$ \\
\hline
\end{tabular}
\caption{\label{tab:table-intro}Number of centers needed to achieve $\Theta(1)$-approximation compared to $opt_k$ with constant $k$. The size of the input is $n$. In light blue, the contribution of this work. See the text for more details. }
\end{center}
\end{table}

\subsubsection{Summary contributions}\label{subsec:summary_contributions}
The contributions of the paper are summarized as follows. 
\paragraph{Complexity measure.} 
We introduce a new measure, $\seq_k(X)$, to identify X's complexity when clustering it in the online no-substitution setting and \emph{arbitrary} order. It helps to quantify the number of points in $X$ needed to be taken as centers. 
The new measure is the longest sub-series in $X$ such that any point is far from all points preceding it.  

\paragraph{Algorithm.} We design a new random algorithm in the online no-substitution setting. It uses, as a subroutine, an approximated clustering algorithm $A$ for the offline setting. For each new point $x_t$, the algorithm uses $A$ to cluster all points observed so far. Then it performs a slight modification to $A$'s clustering by merging all clusters that are close to $x_t$. The new cluster that $x_t$ is in, $x_t\in C^t_i$, will determine the probability of taking $x_t$ as center.  See more details in Section~\ref{sec:main_algorithm}.

\paragraph{Provable guarantees.}
We prove that when running the algorithm on data $X$ with an $\alpha$-approximation offline clustering algorithm, it takes $O(\seq_k(X)\log(n) k\log k )$ centers and is an $O(\alpha k^3)$-approximation.
 We show that the number of centers needed to achieve an $\alpha$-approximation is at least $\Omega\left(\frac{\seq_k(X)}{k \log (n\alpha)}\right)$.
Specifically, suppose $\seq_k(X) = \polylog(n)$ and $k,\alpha$ are constants. In that case, the number of centers is lower and upper bounded by $\polylog(n)$, which is, up to a polynomial, similar to the random order case. 

\paragraph{Applications.}
We prove that if the data $X$ is sampled from a ``natural" distribution, such as a mixture of Gaussians, then the new measure, $\seq_k(X)$, is equal to $O(k^2\log(n))$. Together with our provable guarantees, this implies that for data generated from ``natural" distributions, our new algorithm takes only $\poly(k\log(n))$ centers and is a $\poly(k)$-approximation. 

%\subsection{Paper organization}

\subsection{Related work}
\paragraph{Online no-substitution setting.} Several works \cite{liberty2016algorithm,moshkovitz2019unexpected, hess2020sequential} designed algorithms in the online no-substitution setting. The works  
\cite{hess2020sequential,moshkovitz2019unexpected} assumed the order is random, and \cite{hess2020sequential,liberty2016algorithm} assumed the data or the aspect ratio is bounded. Both assumptions simplify the problem. 
In this paper, we explore the case where the order is arbitrary, and data is unbounded. This paper shows that the number of centers is determined by $d=\seq_k(X)$. If the aspect ratio is small, then so is $d$, which implies similar results to \cite{liberty2016algorithm}. However, notably, small $d$ does not force the entire data to be bounded or have a small aspect ratio. Thus we can handle cases previous works could not.

\paragraph{Online facility location.} Meyerson \cite{meyerson01} introduced the online variant of facility location. 
Demands arrive one after another and are assigned to a facility. 
Throughout the run of the algorithm, a set of facilities $F$ is maintained. %throughout the run. 
Upon arrival of each demand $p$ there is a choice between (1) instant cost, $d(p,\ell)$, to $p$'s closest location $l\in F$ (2) open a new facility. Opening a facility is irreversible. The total cost is $|F|+\sum_p d(p,\ell).$
 Several variants of this problem were investigated (e.g., \cite{fotakis2011online, lang18,feldkord18}).
One of the main differences between the online facility location and the online no-substitution setting is the term that we try to minimize: either the sum of $|F|+\sum_p d(p,\ell)$ or both terms separately. This seemingly small difference has a significant effect. For example, suppose $\sum_p d(p,\ell)$ is of the order of $n$. The online facility location algorithm can take $|F|=n$, permitting the trivial solution of opening a facility at each location. On the other hand, in the no-substitution setting, $|F|$ should be small. 
 
\paragraph{Streaming with limited memory.} There is a vast literature on algorithms in the streaming setting where the memory is limited  \cite{muthukrishnan2005data, aggarwal2007data}. 
One line of works uses coresets \cite{har2004coresets,phillips2016coresets,feldman2020core}, where a weighted subset, $S$, of the current input points $X_t$ is saved such that a $k$-means solution to $S$ has similar cost as $X_t$.
Another line of work saves a set of candidate centers in memory \cite{guha2000liadan,charikar03, shindler11,guha03,ailon09}.
A different line of work had assumptions on the data like well-separated clusters \cite{braverman11,ackerman14,raghunathan2017learning}.
Notably, in the setting of streaming with limited memory, decisions can be revoked, unlike this paper's requirement. Enabling a decision change when new information presents itself allows the algorithm to take a smaller number of points as centers.  

\section{Preliminaries}

\paragraph{Notation.}
For convenience, we will denote the optimal cost with $k$ centers as $\c_k(X) = \c(\o_k(X))$ and the cost of centers $C$ as $\c(X,C)$. It will be frequently useful to consider the costs associated with using a single cluster center. Because of this, we let $\c(X, x)$ denote $\c(X, \{x\})$
%$\c(\{(X, x)\})$ 
and $\c(X)$ denote $\c_1(X)$. 

\paragraph{$1$-means clustering.} The optimal center of a cluster is its mean. The following well-known lemma will prove useful in our analysis.

\begin{lem}[center-shifting lemma]\label{lem:center_change}
Let $X \subset \R^d$ be a finite set and let $\mu = \frac{1}{|X|}\sum_{x \in X} x$. For any $x \in \R^d$, we have $\c(X, x) = \c(X) + |X|d(x, \mu)^2$. 
\end{lem}

\section{Online clustering using \texorpdfstring{$(\alpha,k)$}{Lg}-sequences}\label{sec:alpha_k_seq}

In this section, we introduce $(\alpha, k)$-sequences, which are the principle object from which $\seq_k(X)$ is constructed. %We begin with the following helpful definition.

% \subsection{$(\alpha, k)$-sequences}

% \begin{defn}
% For a set of points $X = \{x_1, x_2, \dots, x_n\}$, its $l$-diameter, denoted $\diam_l(X)$, is defined as the smallest real number for which there exists a partitioning of $X$ into $l$ sets each of which have diameter at most $\diam_l(X)$. In short, $$\diam_l(X) = \min_{X_1 \cup X_2\dots X_l = X} \max_{1 \leq i \leq l}\diam(X_i).$$
% \end{defn}

%Using this, we define an $(\alpha, k)$-sequence, which will be the central object of this section. 

\begin{defn}\label{dfn:alpha-seq}
Let $\alpha > 1$. An $(\alpha, k)$-sequence is an ordered sequence of points $x_1, y_2, \dots, x_m$ such that for $1 < j \leq m$, $$ \min_{1 \leq i \leq {j-1}} d(x_i, x_j) > \alpha \emph{\diam}_{k-1}(\{x_1, x_2, \dots, x_{j-1}\}).$$ 
\end{defn}

An $(\alpha, k)$-sequence can be thought of a ``worst-case" sequence for an online clustering algorithm. For each successive point, the algorithm is strongly incentivized to select that point, since doing so would incur a large cost compared to the cost for any other point. 

The exact value of $\alpha$ in Definition~\ref{dfn:alpha-seq} is insignificant, it solely essential that $\alpha>1$.  Converting an $(\alpha,k)$-sequence into a $(\beta, k)$-sequence for some $\beta \neq \alpha$ has only slight effect on the length of sequence. 
We formalize this in the following lemma, which will also play a key role in providing lower bounds on the number of centers an approximation online algorithm must choose. See section  section \ref{sec:lower_bounds} for the proof. 

\begin{lem}\label{lem:cahnge_factor_in_seq} \label{lem_diameter}
Suppose $1 < \alpha < \beta$. Let $x_1, x_2, \dots, x_n$ be an $(\alpha, k)$-sequence. Then there exists a sub-sequence of length at least $\big\lfloor \frac{n}{2k\log_\alpha \beta }\big\rfloor$ that is a $(\beta, k)$-sequence.
\end{lem}

Lemma~\ref{lem_diameter} shows that we can construct $(\alpha, k)$-sequences from each other for different values of $\alpha$ at a relatively small reduction in sequence length. Because of this, the complexity measure we suggest, $\seq_k$, essentially fixes $\alpha = 2$. 
% \begin{defn}
% For any set of points $X$, we let $\seq_k(X)$ denote the largest subset of $X$ that can be ordered to form a $(2, k)$-sequence. 
% \end{defn}
Throughout this paper, we express both upper and lower bounds on the number of centers an online algorithm chooses given input $X$ through $\seq_k(X)$.

\section{An Algorithm for Online Clustering}\label{sec:main_algorithm}

In this section, we present and analyze our new online algorithm. At a high level, at each time step $t$, the algorithm computes an offline $k$-clustering of the points received so far $x_1, x_2, \dots, x_t$. It then maximally merges clusters containing $x_t$ (without increasing the total cost by more than a constant factor), and finally chooses $x_t$ with probability inversely proportional to final merged cluster containing it. 

The intuition here is that points end up in small merged clusters are more likely to be difficult to cluster with other points, while points in large ones are more likely to be easily clustered with others. Correspondingly, points in small clusters are chosen with high probability while points in large clusters are chosen with small probability.

To reduce the number of centers chosen, the algorithm modifies the offline clustering of $x_1, \dots, x_t$ by maximally combining clusters so that $x_t$ is contained in as large a cluster as possible without increasing the total cost too much. After doing so, it then chooses $x_t$ as outlined above. 

%\{maybe we can think about a cool name for the algorithm + will have "arbitrary order" in the name as this is the main contribution of the algorithm}

\begin{algorithm}
\caption{Online clustering algorithm for arbitrary order}
\label{alg:1}
\SetAlgoLined
\KwIn{A stream of points $X = \{x_1, x_2, \dots, x_n\}$, a desired number of clusters $k$, an offline clustering algorithm $A$}
\KwOut{A set of cluster centers $C \subseteq X$}
$C \leftarrow \{\}$\;

\For{$t \in \{1, 2, 3, \dots, n\}$}{
$\C_t \leftarrow A(\{x_1, x_2, \dots, x_t\}, k)$ \;

$\C_t = \{(C_t^1, c_t^1), (C_t^2, c_t^2), \dots, (C_t^k, c_t^k)\}$ \;

without loss of generality $d(c_t^1, x_t) \leq d(c_t^2, x_t) \leq \dots \leq d(c_t^k, x_t)$ \;

without loss of generality $c_t^i$ is the optimal cluster center for $C_t^i$ \;

$\v_t \leftarrow \max(\{i: \sum_{j = 1}^{i} |C_t^j|d(c_t^j, c_t^i)^2 \leq 100 \c(C_t^1, C_t^2, \dots, C_t^k)\})$\;

$s_t \leftarrow \sum_{j=1}^{\v_t} |C_t^j|$\;

with probability $\frac{20k\log k}{s_t}$, $C = C \cup \{x_t\}$ \;
}
return $C$\;
\end{algorithm}

Let $v_t$ be defined as in Algorithm \ref{alg:1}. It will prove useful to think of the center  $c_{t}^{\v_t}$ as replacing the centers $\{c_t^1, c_t^2, \dots, c_t^{\v_t}\}$. In essence, Algorithm \ref{alg:1} clusters $C_t^1, C_t^2, \dots, C_t^{\v_t}$ with center $c_{t}^{\v_t}$, and clusters $C_t^j$ with $c_t^j$ for $j > \v_t$. To this end, we let $c_t(x)$ denote the center that $x$ is clustered with at time $t$. In particular, for any $x \in \{x_1, x_2, \dots x_t\}$, let $$c_t(x) = \begin{cases} c_t^{\v_t} & x \in C_t^{j}, j \leq \v_t \\c_t^j & x \in C_t^j, j > \v_t \end{cases}.$$

In the following lemma, we show that this ``new" clustering still keeps the total cost at each time $t$ bounded. That is, the new clustering used by Algorithm \ref{alg:1} is an $O(1)$-approximation to the optimal $k$-clustering at every time $t$.

\begin{lem}\label{lem:approx_upper_bound}
For any $1 \leq t \leq n$, $\sum_{i=1}^t d(x_i, c_t(x_i))^2 \leq 101\alpha \c_k(\{x_1, x_2, \dots, x_t\}).$
\end{lem}

\begin{proof}
Observe that for $j \leq \v_t$, we have  $\sum_{x_i \in C_t^j} d(x_i, c_t(x_i))^2 = \c(C_t^j, c_t^{\v_t})$, while for $j > \v_t$, we have $\sum_{x_i \in C_t^j} d(x_i, c_t(x_i))^2 = \c(C_t^j, c_t^{j}) = \c(C_t^j)$. Summing these equations over all $j$ and applying the center-shifting lemma, we have
\begin{equation*}
\begin{split}
\sum_{i=1}^t d(x_i, c_t(x_i))^2 &= \sum_{j = 1}^{k} \sum_{x_i \in C_t^j} d(x_i, c_t(x_i))^2 \\
&= \sum_{j = 1}^{\v_t} \c(C_t^j, c_t^{\v_t}) + \sum_{j = \v_t + 1}^k \c(C_t^j) \\
&= \sum_{j=1}^{\v_t} \c(C_t^j) + |C_t^j|d(c_t^j, c_t^{\v_t})^2 + \sum_{j = \v_t+1}^k \c(C_t^j) \\
&= \c(\C_t) + \sum_{j = 1}^{\v_t} |C_t^j|d(c_t^j, c_t^{\v_t})^2 .
\end{split}
\end{equation*}

By the definition of $\v_t$, we have that $\sum_{j = 1}^{\v_t} |C_t^j|d(c_t^j, c_t^{\v_t})^2 \leq 100\c(\C_t)$. Substituting this, implies that $\sum_{i=1}^t d(x_i, c_t(x_i))^2 \leq 101 \c(\C_t)$. Since $A$ is an approximation algorithm with approximation factor $\alpha$, we have that $\c(\C_t) \leq \c_k(\{x_1, x_2, \dots, x_t\})$, which implies the lemma.
\end{proof}

In section \ref{sec:approx_factor}, we show that this algorithm is an online $\poly(k)$-approximation algorithm, and in section \ref{sec:center_complexity}, we analyze the expected number of centers chosen.

Theorem \ref{thm:asp_ratio_bound} implies that although there exist input sequences $X$ for which any online approximation algorithm must take many centers (i.e. $\Omega(n)$), for input sequences $X$ that are sampled from some well-behaved probability distribution, it is possible to do substantially better \textit{regardless of the order $X$ is presented in}. 

\subsection{Approximation factor analysis}\label{sec:approx_factor}

The analysis of our algorithm hinges around the following known observation: taking a point at random from a single cluster yields a decent approximation for a cluster center. For completeness, we formalize this observation with the following lemma.  

\begin{defn}\label{def:good}
Let $S$ be any set of points, and let $G \subset S$ be defined as $G = \{x: \c(S, x) \leq 3\c(S)\}$. We refer to the points $g \in G$ as \emph{good} points.
\end{defn}

\begin{lem}\label{lem:good}
Let $S$ be any set of points, and let $G \subset S$ be the good points in $S$. Then $|G| \geq \frac{|S|}{2}$.
\end{lem}

\begin{proof}
Let $n = |S|$ and $\mu = \frac{1}{|S|}\sum_{s \in S} s$ denote the mean of $S$. For $x \notin G$, $\c(S, x) > 3\c(S)$.  By Lemma \ref{lem:center_change}, $\c(S, x) = \c(S) + nd(x, \mu)^2$. Therefore $d(x, \mu)^2 > \frac{2\c(S)}{n}$. However, $\sum_{x \in S} d(x, \mu)^2 = \c(S)$ by definition. Thus we have $$\c(S) \geq \sum_{x \notin G} d(x, \mu)^2 > \frac{2\c(S)|S \setminus G|}{n}.$$ This implies $\frac{|S \setminus G|}{n} < \frac{1}{2}$, which means $|G| > \frac{n}{2}$, as desired.
\end{proof}

Lemma~\ref{lem:good} implies that if points from $S$ are independently selected with probability $\Theta\left(\frac{1}{|S|}\right)$, then it is likely that some point $g \in G$ will be selected. We will use this idea to argue that Algorithm \ref{alg:1} selects good points from each cluster in $\o_k(S)$ with high probability.

\begin{thm} \label{thm:approx}
Let $A$ be an offline clustering algorithm with approximation factor $\alpha$. Suppose running Algorithm \ref{alg:1} on $X, k, A$ returns a set of centers $C$. Then with probability $0.9$ over the randomness of Algorithm \ref{alg:1}, $$\c(X,C) \leq 1358\alpha k^3\c_k(X).$$
\end{thm}

Before giving a proof, we first describe our proof strategy and give some helpful definitions and lemmas. 

Let $C_*^1, C_*^2, \dots, C_*^k$ denote the optimal $k$-clustering of $X$, and let $G^1, G^2, \dots, G^k$ denote the sets of good points in each cluster. We also let $|X| = n$.

Our proof strategy is the following. We will first show that there exist clusters $C_*^i$ for which we are likely to choose some center $g \in G^i$. Therefore, for these clusters, we have a $3$-approximation of the optimal cost. For the remaining clusters, we will argue that clusters that are not likely to have a good point chosen must be ``close" to some other cluster. We will then conclude that the good points that we have already selected will serve as an approximation for \textit{all} cluster $C_*^i$, which implies that the total cost is bounded by some constant times $\c_k(X)$.

We begin with the first step. Using the notation from Algorithm \ref{alg:1}, we let $s_t = |\cup_{i=1}^{\v_t} C_t^i|$. For cases in which we don't explicitly note the index $t$, we will let $s(x)$ denote the same thing (i.e. $s(x_t) = s_t$). Here, $s(x)$ can be thought of as the set of points clustered with $x$ (or near $x$).

\begin{lem}\label{lem:easy_clusters}
Fix any $1 \leq i \leq k$. Suppose that for at least half the points $g \in G^i$, $s(g) \leq k|C_*^i|$. Then with probability at least $1 - \frac{1}{k^5}$, Algorithm \ref{alg:1} selects some $g \in G^i$. 
\end{lem}
\begin{proof}
For any $g$ with $s(g) \leq k|C_*^i|$, Algorithm \ref{alg:1} selects $g$ with probability at least $\frac{20k\log k}{k|C_*^i|} \geq \frac{10\log k}{|G^i|}$. Since there are at least $\frac{|G^i|}{2}$ such points $g$, and since each point is selected with an independent coin toss, Algorithm \ref{alg:1} selects \textit{no} such points with probability at most $\left(1 - \frac{10\ln k}{|G^i|}\right)^{\frac{|G^i|}{2}}$. By standard manipulations, we have $$\left(1 - \frac{10\ln k}{|G^i|}\right)^{\frac{|G^i|}{2}} \leq e^{-5\ln k} = \frac{1}{k^5},$$ as desired. 
\end{proof}

\begin{figure}
\centering
\includegraphics[scale=0.4]{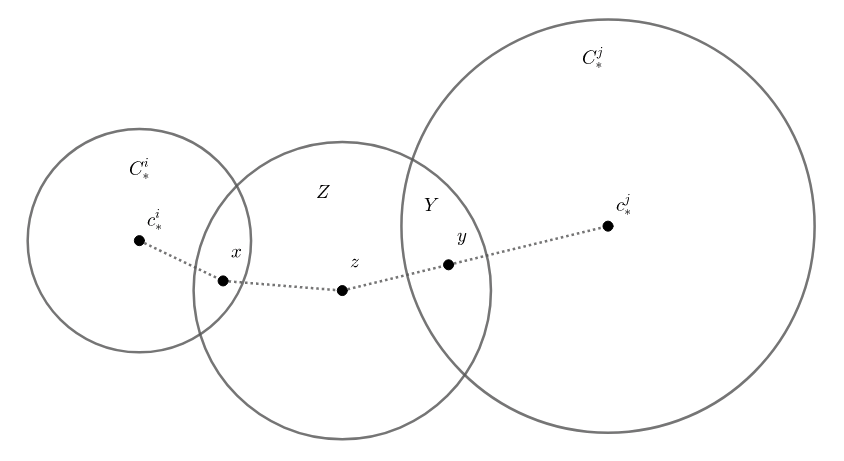}
\begin{tabular}{r@{: }l r@{: }l}
$C_*^i$ & $i$th optimal cluster & $c_*^i$ & center of $C_*^i$\\
$C_*^j$ & $j$th optimal cluster & $c_*^j$ & center of $C_*^j$\\
$Z$& $\cup_{j=1}^{\v_t} C_t^i$ & $z$& $c_t^{\v_t}$ \\
$Y$ & $C_*^j \cap Z$ & $y$ & $\frac{1}{|Y|}\sum_{y' \in Y} y'$
\end{tabular}
\caption{Proof idea of Lemma \ref{lem:hard_clusters}. We derive upper bounds on all the dotted lines.}
\end{figure}

Next we do the second step, in which we handle clusters for which we are not likely to choose some good point $g \in G^i$. 
\begin{lem} \label{lem:hard_clusters}
Fix any $1 \leq i \leq k$. Suppose that for strictly less than half the points $g \in G^i$, $s(g) \leq k|C_*^i|$. Then there exists $j$ with $|C_*^j| > |C_*^i|$ such that $$|C_*^i|d(c_*^i, c_*^j)^2 \leq 1254\alpha\c_k(X),$$ where $c_*^i$ and $c_*^j$ denote the centers of $C_*^i$ and $C_*^j$ respectively. 
\end{lem}

\begin{proof}
Let $t$ be the last time such that $x_t$ is a good point from $C_*^i$ (i.e. $x_t \in G^i$), and such that $s_t > k|C_*^i|$. By assumption, for strictly more than $\frac{|G^i|}{2}$ points $x$ in $G^i$, $s(x) > k|C_*^i|$.  This implies that $|\{x_1, x_2, \dots, x_t\} \cap G^i| \geq \frac{|G^i|}{2}.$ The idea is to analyze the algorithm at time $t$.

Let $Z$ denote the combined cluster that $x_t$ is assigned to at time $t$ by Algorithm \ref{alg:1}, that is $Z = \cup_{i=1}^{\v_t} C_t^i$ and $z = c_{\v_t}^t$. In particular, we have that $s_t = |Z|$. Because $|Z| > k|C_*^i|$, there exists $1 \leq j \leq k$ such that $|C_*^j \cap Z| > |C_*^i|$. We claim that for this value of $j$, $C_*^j$ satisfies the desired properties in the lemma.  

Our goal is to find an upper bound on $d(c_*^i, c_*^j)$. To do this, we will find bounds on $d(c_*^i, x_t)$, and  $d(x_t, z),$ and $d(z, c_*^j)$, and then use the triangle inequality. For bounding $d(z, c_*^j)$ in particular, we will consider the intersection of $Z$ and $C_*^j$ which we denote as $Y = C_*^j \cap Z$. We let $y = \frac{1}{|Y|}\sum_{y' \in Y} y'$ be the average of all points in $Y$, and will subsequently bound $d(z, c_*^j)$ by bounding $d(z,y)$ and $d(y, c_j^*)$. 

We will argue this by finding bounds on $d(c_*^i, x), d(x, z), d(z, y),$ and $d(y, c_*^j)$, and then using the triangle inequality. Figure 1 gives a picture that summarizes this argument. We will derive upper bounds on all the dotted lines. 

\paragraph{Claim 1:} $|C_*^i|d(c_*^i, x_t)^2 \leq 2\c_k(X)$.

Since $x_t \in G^i$, we have $\c(C_*^i, x_t) \leq 3\c(C_*^i)$. Therefore, we have $\c(C_*^i) + |C_*^i|d(c_*^i, x_t)^2 \leq 3\c(C_i^*).$ Subtracting $\c(C_i^*)$ from both sides and substituting $\c(C_*^i) \leq \c_k(X)$ gives the result.

\paragraph{Claim 2:} $|C_*^i|d(x_t, z)^2 \leq 526\alpha \c_k(X)$. 

Let $G \subset G^i$ denote the set of all good points present at time $t$ in $C_*^i$. In particular, we let $G = \{x_1, x_2, \dots, x_t\} \cap G^i$. Recall that from the definition of $t$, we have $|G| \geq \frac{|G^i|}{2} \geq \frac{|C_*^i|}{4}$. 

For any $x' \in G$ let $z'$ be its closest center in $\{c_t^{\v_t}, \dots, c_t^k\}$. In particular, we have $z' = \argmin_{c_t^j \in \{c_t^{\v_t},\ldots, c_t^k\}} d(x', c_t^j)$. The key observation is that $x_t$ must be closer to $z$ than it is to $z'$ (from the definition of $z$). Applying the triangle inequality, it follows that $d(x_t, z) - d(x_t, x') \leq d(x', z')$. 

Observe that $x_t$ and $x'$ are both good points, and consequently by the argument in Claim~1, $|C_*^i|d(c_*^i, x_t)^2  \leq 2\c_k(X)$ and $|C_*^i|d(c_*^i, x')^2 \leq 2\c_k(X)$. Thus by applying the triangle inequality again, $|C_*^i|d(x_t, x')^2 \leq 8\c_k(X)$. 

Suppose that $d(x_t, z) \leq d(x_t, x')$. Then this implies $|C_*^i|d(x_t, z)^2 \leq 8\c_k(X)$ which implies the claim. In the other case, we assume $d(x_t, z) \geq d(x_t, x')$, which implies that $(d(x_t, z) - d(x_t, x'))^2 \leq d(x', z')^2.$ The idea now is to sum this equation over all $x' \in G$ and substitute $|C_*^i|d(x_t, x')^2 \leq 8\c_k(X)$ to get that $$|G|\left( d(x_t, z) - \sqrt{\frac{8\c_k(X)}{|C_*^i|}} \right)^2 \leq \sum_{x' \in G} d(x', z')^2.$$

Since $G \subset \{x_1, x_2, \dots, x_t\}$, it follows that $\sum_{x' \in G} d(x', z')^2$ is at most the cost of assigning each $x_i$ to their nearest center in $\{c_t^{\v_t}, c_t^{\v_t + 1}, \dots, c_t^k\}$. This is upper bounded by Lemma \ref{lem:approx_upper_bound}, implying that $\sum_{x' \in G} d(x', z')^2 \leq 101\alpha \c_k(X).$ Substituting this and observing that $|G| \geq \frac{|C_*^i|}{4}$, we have that
\begin{equation*}
\begin{split}
|C_*^i|d(x_t, z)^2 &\leq |C_*^i|\left ( \sqrt{\frac{101\alpha\c_k(X)}{|G|}} +\sqrt{\frac{8\c_k(X)}{|C_*^i|}}\right )^2 \\
&\leq |C_*^i|\left ( \sqrt{\frac{404\alpha\c_k(X)}{|C_*^i|}} +\sqrt{\frac{8\c_k(X)}{|C_*^i|}}\right )^2 \\
&\leq \left (\sqrt{404\alpha \c_k(X)} + \sqrt{8\c_k(X)} \right )^2 \\
&\leq 526\alpha\c_k(X),
\end{split}
\end{equation*}
as desired.

\paragraph{Claim 3:} $|C_*^i|d(z, y)^2 \leq 101\alpha \c_k(X)$.

Observe that the cost incurred by $Y$ at time $t$ by the modified offline clustering (with centers $\{c_t^{v_t}, c_t^{v_t+1}, \dots, c_t^k\}$) is $\c(Y, z) = \c(Y) + |Y|d(z, y)^2$. By Lemma \ref{lem:approx_upper_bound}, this cost is at most $101\alpha \c_k(X)$, and since $|Y| > |C_*^i|$, the result follows.

\paragraph{Claim 4:} $|C_*^i|d(y, c_*^j)^2 \leq \c_k(X)$.

Since $Y \subseteq C_*^j$, the cost $\c(Y, c_*^j)$ at time $t$ by the modified offline clustering is at most $\c(C_*^j)$. Because $|Y| > |C_*^i|$, we have that $|C_*^i|d(y, c_*^j)^2 + \c(Y) \leq \c(C_*^j) \leq \c_k(X),$ which implies the result. 

\paragraph{Putting it all together.} Armed with all $4$ of our claims, we can prove the lemma using the triangle inequality, 
\begin{equation*}
\begin{split}
|C_*^i|d(c_*^i, c_*^j)^2 &\leq |C_*^i|(d(c_*^i, x_t) + d(x_t, z) + d(z, y) + d(y, c_*^j))^2 \\
&\leq |C_*^i|\left ( \sqrt{\frac{2\c_k(X)}{|C_*^i|}}+\sqrt{\frac{526\alpha\c_k(X)}{|C_*^i|}}+\sqrt{\frac{101\alpha\c_k(X)}{|C_*^i|}}+\sqrt{\frac{\c_k(X)}{|C_*^i|}}\right)^2 \\
&\leq (\sqrt{2\alpha} + \sqrt{526\alpha} + \sqrt{101\alpha} + \sqrt{\alpha})^2\c_k(X) \\
&\leq 1254\alpha\c_k(X).
\end{split}
\end{equation*}
\end{proof}

We now complete the proof of Theorem \ref{thm:approx}. 

\begin{proof}
(Theorem \ref{thm:approx}) Let $T$ denote the set of all $i$ such that for at least half the points $g \in G^i$, $s(g) \leq k|C_*^i|$. Although $T$ is a random set, its randomness only stems from the randomness in the approximation algorithm $A$. Crucially, the set $T$ is independent from the results of the random choices that ultimately determine which elements of $X$ we select in $C$. Thus, in this proof we will treat $T$ and $s(g)$ as fixed entities, and evaluate all probabilities over the randomness from the random choices.

Using a union bound along with  Lemma \ref{lem:easy_clusters}, we see that with probability at least $1 - \frac{k}{k^5} \geq 0.9$, we will choose some $g \in G^i$ for all $i \in T$. Recall that $C$ denotes the output of Algorithm \ref{alg:1}. Therefore, with probability $0.9$, for all $i \in T$, we have $$\c(C_*^i, C) \leq 3\c(C_*^i).$$

Next, select any $1 \leq i \leq k$ with $i \notin T$. Although applying lemma \ref{lem:hard_clusters} may not result in $j \in T$, it will result in $j$ such that $|C_*^j| > |C_*^i|$. Therefore, applying this lemma at most $k$ times will continually result in increasingly large sets $|C_*^j|$. Since such a $C_*^j$ is guaranteed to exist, this must terminate in some $j \in T$.  

Therefore, using the triangle inequality in conjunction with Cauchy Schwarz we see that there exists $j \in T$ such that $$|C_i^*|d(c_*^i, c_*^j)^2 \leq 1254\alpha k^2\c_k(X).$$ Let $g$ be a good point with $g \in G^{j}$. By the same argument given in Claim~1, $|C_*^i|d(g, c_*^j)^2 \leq |C_*^j|d(g, c_*^j)^2 \leq 2\c_k(X)$. Therefore, by the triangle inequality $$|C_*^i|d(c_*^i, g)^2 \leq |C_*^i|\left(d(c_*^i, c_*^j) + d(c_*^j, g)\right)^2 \leq 1357\alpha k^2\c_k(X).$$ By Lemma \ref{lem:center_change}, this implies $$\c(C_*^i, g) = |C_i^*|d(c_*^i, g)^2 + \c(C_*^i) \leq1358\alpha k^2\c_k(X).$$ As shown earlier, Algorithm \ref{alg:1} selects some $g \in G^{j}$ with probability at least $0.9$ for all $j \in T$. Therefore, with probability $0.9$, Algorithm \ref{alg:1} outputs $C$ such that $$\c(X, C) \leq \sum_{i = 1}^k \c(C_*^i, C) \leq \sum_{i = 1}^k 1358\alpha k^2\c_k(X) = 1358\alpha k^3\c_k(X)$$ as desired. 

\end{proof}

\subsection{Center complexity analysis}\label{sec:center_complexity}

We now bound the expected number of centers outputted by Algorithm \ref{alg:1}.

\begin{thm}\label{thm:center_complexity}
Let $A$ be an offline $\alpha$-approximation algorithm. If Algorithm \ref{alg:1} has output $C$, then $$\E[|C|] \leq 160k\log k \seq_{k}(X)(\log n+1).$$
\end{thm}

Before proving Theorem \ref{thm:center_complexity}, we introduce some useful definitions and lemmas. 

As before, we let $X = \{x_1, x_2, \dots, x_n\}$ denote our input, and let $s_t$ be as defined in Algorithm \ref{alg:1}. Observe that $x_i$ is selected as a center by our algorithm with probability $\frac{20k\log k}{s_i}$. Therefore, the expected number of centers satisfies $$\E[|C|] = \sum_{i=1}^n \frac{20k\log k}{s_i}.$$ Our goal will be to show an upper bound on this expression. To do so, we will need the following constructions. For any $1 \leq t \leq n$, define the following.

\begin{enumerate}
	\item Let $P_t \subseteq \{x_1, x_2, \dots, x_t\}$ denote $P_t = C_t^1 \cup C_t^2 \dots \cup C_t^{\v_t}$. This represents all elements that were clustered with $x_t$ by Algorithm 1 at time $t$.
	\item For any $x_t$, we let $r_t = r(x_t) = d(x, c_t^{\v_t + 1}).$ Thus $r(x_t)$ represents the distance from $x_t$ to the closest center at time $t$ that is not used for clustering $P_t$. 
	\item Let $Q_t$ denote the set of all $x \in \{x_1, x_2, \dots x_t\} \setminus P_t$ such that $d(x, c_t(x)) \geq \frac{1}{5}r(x_t)$, where $c_t(x)$ denote the cluster center $c_j^t$ that $x$ is assigned to at time $t$ by offline clustering algorithm $A$.
\end{enumerate}
 
Finally, we will use these sets $P_t, Q_t$ for $1 \leq t \leq n$ to construct one directed graph $G$ as follows. Let $G$ have vertex set $\{1, 2, 3, \dots, n\}$ and every vertex $j$ have edges pointed to all $i$ for which $x_i \in P_j$ or $x_i \in Q_j$. In particular, the set of edges in $G$ denoted $E(G)$ satisfies $$E(G) = \{(j, i): x_i \in (P_j \cup Q_j)\}.$$
Our strategy will be show the following:
\begin{enumerate}
	\item Any independent set $I \subset G$ forms a $(2, k)$-sequence. 
	\item $G$ has an independent set of size at least $\frac{1}{8(\log n+1)}\sum_{i=1}^n \frac{1}{s_i}$. 
\end{enumerate}

These two observations will imply that $$\E[|C|] = \sum_{i=1}^n \frac{20k\log k}{s_i} \leq 160k\log k \seq_{k}(X)(\log n+1), $$ which is the desired result. We now verify these observations with the following lemmas. 

\begin{lem}\label{lem:seq_construct}
If $i_1 \leq i_2, \dots \leq i_r$ is an independent set in $G$, then $x_{i_1}, x_{i_2}, \dots, x_{i_r}$ form a $(2,k)$-sequence. 
\end{lem}

\begin{proof}
Fix any $1 < s \leq r$, and for convenience let $t = i_s$ and let $I_s = \{x_{i_j}: 1 \leq j < s\}$ denote the set of points before $x_{i_s}$. The key observation is that $C_t^{\v_t + 1}, C_t^{\v_t + 2}, \dots, C_t^k$ partition $I_s$ into at most $k-1$ sets. For any $x_{i_j} \in I_s$, $x_{i_j} \notin P_t \cup Q_t$. Therefore, $d(x_{i_j}, c(x_{i_j})) < \frac{1}{5}r(x_t)$. 

This means we have a partitioning of $I_s$ into $k-1$ sets each of which has diameter strictly less than $\frac{2}{5}r(x_t)$. Meanwhile the distance from $x_t$ to the closest point in $I_s$ is at least $r(x_t) - \frac{1}{5}r(x_t) = \frac{4}{5}r(x_t).$ This is strictly more than double the $(k-1)$-diameter of $I_s$. Since $s$ was arbitrary, we see that the precise condition for a $(2,k)$-sequence holds as desired.
\end{proof}

\begin{lem}\label{lem:degree}
For any $t$, $|P_t \cup Q_t| \leq 2s_t-1$. Thus the vertex $t$ has out-degree at most $2s_t-1$ in $G$. 
\end{lem}

\begin{proof}
Recall that $s_t$ was defined as $s_t = |C_t^1 \cup C_t^2 \cup \dots \cup C_t^{\v_t}| = |P_t|$. Therefore, it suffices to show that $|Q_t| \leq s_t-1$. 

Assume towards a contradiction that $|Q_t| \geq s_t$. Let $L_t$ denote the cost of the original clustering chosen by Algorithm $A$ at time step $t$. That is $L_t = \sum_{i=1}^k \c(C_t^i).$ By definition, each element in $Q_t$ incurs a cost of at least $\frac{1}{25}r_t^2$. Therefore, we have that $L_t \geq \frac{s_tr_t^2}{25}.$

Next, we bound the cost of assigning $C_t^1, C_t^2, \dots, C_t^{\v_t}, C_t^{v_t + 1}$ to $c_t^{v_t + 1}$. To do so, observe that $d(x_t, c_t^i) \leq r_t$ for any $i \leq \v_t + 1$. This is because the centers $c_t^i$ are arranged in increasing order by their distance from $x_t$. Therefore, by the triangle inequality, for any $1 \leq i \leq v_t$, $d(c_t^i, c_t^{v_t+1}) \leq 2r_t$. This implies that 
\begin{equation*}
\begin{split}
\sum_{i=1}^{\v_t}|C_t^i|d(c_t^i, c_t^{v_t+1})^2  &\leq  \sum_{i=1}^{\v_t}|C_t^i|4r_t^2 \\
&\leq 4s_tr_t^2 \\
&\leq 100L_t,
\end{split}
\end{equation*}
with the last inequality holding because $L_t \geq \frac{s_tr_t^2}{25}.$ However, this contradicts the maximality of $\v_t$, which implies that our assumption was false and $|Q_t| < s_t$ as desired.
\end{proof}

Next, we will find a lower bound on the size of the largest independent set in $G$. Our main tool for doing so is Turan's theorem which we review in the following theorem. For completeness, we also include a proof.

\begin{thm}[Turan's theorem] Let $H$ be an undirected graph with average degree $\Delta$. Then there exists an independent set in $H$ consisting of at least $\frac{|H|}{\Delta + 1}$ vertices, where $|H|$ denotes the number of vertices in $H$.
\end{thm}

\begin{proof}
Let $|H| = m$ and let our vertices be labeled $1, 2, 3, \dots m$. Let vertex $i$ have degree $\Delta_i$. Take a random ordering of the vertices in $H$. Proceed through the vertices in this order and select a vertex if and only if none of its neighbors has already been selected. At the end of this process, we are clearly left with an independent set $I$. The probability that vertex $i$ is included in $I$ is precisely $\frac{1}{\Delta_i + 1}$, since $i$ will be chosen if and only if it appears before its $\Delta_i$ neighbors. Thus, by linearity of expectation, we have that $\E[|I|] = \sum_{i=1}^m \frac{1}{\Delta_i + 1}$. Thus there exists an independent set $I^*$ with size at least $\sum_{i=1}^m \frac{1}{\Delta_i + 1}$. 

To finish the proof, let $f(x) = \frac{1}{x}$, thus $|I^*| \geq \sum_{i=1}^m f(\Delta_i + 1)$. The key observation is that $f$ is a convex function on the interval $(0, \infty)$, and thus by Jensen's inequality, we have that $$|I^*| \geq mf\left(\frac{\sum_{i=1}^m \Delta_i + 1}{m}\right) = mf(\Delta + 1) = \frac{m}{\Delta+1},$$ as desired. 
\end{proof}

\begin{lem}\label{lem:turan}
$G$ has an independent set of size at least $\frac{1}{8(\log n+1)}\sum_{i=1}^n \frac{1}{s_i}$. 
\end{lem}
\begin{proof}
Let $d_m$ denote the outdegree of vertex $m$. Partition the vertices of $G$, $\{1, 2, 3, \dots, n\}$ into sets $S_0, S_1, S_2, \dots, S_{\log n}$ such that $$S_i = \{j: 2^{i-1} \leq d_j < 2^i\}.$$ We let $S_0$ be the set of all vertices with degree $0$. Let $G_i$ denote the subgraph of $G$ induced by $S_i$. The main idea is to use Turan's theorem on each graph $G_i$, and then use the fact that there are $\log n$ graphs $G_i$ to consider. 

Observe that $G_i$ has at most $|S_i|(2^i - 1)$ edges. By considering the undirected version of $G$ (simply drop the orientation of each edge), it follows that the average degree is at most $\frac{|S_i|(2^{i+1} - 2)}{|S_i|} = 2^{i+1} -2.$ Therefore, by Turan's theorem, $G_i$ has an independent set $I_i$ of size at least $\frac{|S_i|}{2^{i+1}}$. As a result, we have that $$\sum_{j \in S_i} \frac{1}{d_j+1} \leq \frac{|S_i|}{2^{i-1}} \leq 4|I_i|.$$ Let $I$ denote the largest independent set of $G$. It follows that $I \geq |I_i|$ for all $i$. Summing the above inequality over all $i$, we see that $$\sum_{j = 1}^n \frac{1}{d_j + 1} \leq 4\sum_{i=0}^{\log n} |I_i| \leq 4|I|(\log n+1).$$ By Lemma \ref{lem:degree}, $d_m \leq 2s_m - 1$. Upon substituting this, the desired result follows. 
\end{proof}

We are now ready to prove Theorem \ref{thm:center_complexity}.
\begin{proof}
(Theorem \ref{thm:center_complexity}) Let $I$ be the largest independent set of $G$. By Lemma \ref{lem:seq_construct}, we have that $|I| \leq \seq_{k}(X)$. By Lemma \ref{lem:turan}, we have that $$\frac{1}{8(\log n+1)}\sum_{i=1}^n \frac{1}{s_i} \leq |I| \leq \seq_{k}(X).$$ Multiplying by $\log n$, we see that $$\E[|C|] = \sum_{i=1}^n \frac{20k\log k}{s_i} \leq 160k\log k \seq_{k}(X)(\log n+1),$$ as desired. 
\end{proof}
\section{Lower Bounds}\label{sec:lower_bounds}

In this section, we prove lower bounds on the number of centers any online algorithm with approximation factor $\alpha$ must take. We first express these bounds in terms of $(\frac{1}{2}\sqrt{n\alpha}, k)$-sequences, and then convert them to bounds involving $\seq_k(X)$ by utilizing Lemma \ref{lem_diameter}. The basic idea is that in an $(\frac{1}{2}\sqrt{n\alpha}, k)$ sequence, the points spread out at an extremely quickly rate. Therefore, each subsequent point must be selected for otherwise it incurs are large cost. 

\begin{proof}[of Lemma~\ref{lem:cahnge_factor_in_seq}]
For any $1 < m \leq n$, let $d_m$ denote the distance from $x_m$ to the closest point preceding it, that is, $d_m = \min_{1 \leq i \leq m-1}d(x_i, x_m).$ 

First, we claim that for any $m > k$, there exists $1 \leq i \leq k-1$ such that $d_m > \alpha d_{m-i}.$ To see this, observe that by the definition of a $( \alpha,k)$-sequence, it is possible to partition $\{x_1, x_2, \dots, x_{m-1}\}$ into $k-1$ sets so that each has diameter strictly less than $d_m/\alpha$. By the pigeonhole principle, at least one of $x_{m-1}, x_{m-2}, \dots, x_{m-k+1}$ must be partitioned into a set with more than $1$ element. Let $m-i$ be this value. Then, it follows that $d_m/\alpha > d_{m-i}$.

Next, by repeatedly applying this claim, starting with $x_n$, we can construct a sequence of points  $x_{i_1}, x_{i_2}, \dots, x_{i_r}$ such that $r \geq \frac{n}{k}$ and $d_{i_j} > \alpha d_{i_{j-1}}$ for all $1 < j \leq r$. Note that this sequence is constructed in reverse order by starting with $i_r = n$, and then setting $i_{r-1} = i_r - i$ where $i$ is the value found by the argument above with $1 \leq i \leq k-1.$

Finally, let $s = \lceil \log_\alpha \beta \rceil$. It follows that for all $j$, $d_{i_j} > \alpha^{s-1} d_{i_{j - s + 1}} > \frac{\beta}{\alpha}d_{i_{j - s + 1}}.$ Using the definition of an $(\alpha, k)$-sequence, we have that for any $j$,
\begin{equation*}
\begin{split}
d_{i_j} &> \frac{\beta}{\alpha} d_{i_{j-s+1}} \\
&> \frac{\beta}{\alpha}\alpha\diam_{k-1}(\{x_1, x_2, \dots, x_{i_{j-s}}\}) \\
&\geq \beta\diam_{k-1}(\{x_{i_{j-s}}, x_{i_{j-2s}}, x_{i_{j-3s}}, \dots \}).
\end{split}
\end{equation*}

By repeatedly setting $j$ to be multiples of $s$, we see that $x_{i_s}, x_{i_{2s}}, \dots, x_{i_{\lfloor r/s\rfloor s}},$ is a $(\beta, k)$-sequence. Thus, all the remains is to bound its length. Substituting $r \geq n/k$, we have that 
$\lfloor \frac{r}{s} \rfloor \geq \big\lfloor \frac{n}{k\lceil \log_\alpha \beta \rceil}\big\rfloor,$ which implies the result. 

\end{proof}

\begin{lem}\label{lem:seq_requirement}
Let $x_1, x_2, \dots, x_n$ be an $(\frac{1}{2}\sqrt{n\alpha}, k)$-sequence. Then the expected number of centers taken by \textbf{any} streaming algorithm that guarantees approximation factor $\alpha$ is at least $0.9n$. %points as centers in expectation.
\end{lem}

\begin{proof}
Consider any $1 < m \leq n$. Let $d_m$ denote the distance from $x_m$ to the closest point preceding it; that is, $d_m = \min_{1 \leq i \leq m-1}d(x_i, x_m)$. The key observation is that if the algorithm doesn't choose $x_m$, then it must pay a cost of at least $d_m^2$ at time $m$ (since $x_m$ must be clustered with some cluster center in $\{x_1, x_2, \dots, x_{m-1}\}$). Because $n$ is not known in advance, any streaming Algorithm must ensure that the cost at all times $t$ is relatively low. We will show that $d_m^2$ is large enough so that failing to pick it will incur a cost at time $m$ that is too high.

Consider the following clustering of $\{x_1, x_2, \dots, x_m\}$. Let $x_m$ be its own cluster, and then cluster $\{x_1, x_2, \dots, x_{m-1}\}$ into $k-1$ clusters each with diameter strictly less than $\frac{2d_m}{\sqrt{n\alpha}}.$ This is possible because of the definition of an $(\frac{1}{2}\sqrt{n\alpha}, k)$-sequence. It follows that each point is clustered with radius at most $\frac{d_m}{\sqrt{n\alpha}}$ in this clustering, and thus the entire cost is strictly less than $\frac{d_m^2}{\alpha}$. As a result, it follows that the total cost is small, that is, $\c_k(\{x_1, x_2, \dots, x_m\}) < \frac{d_m^2}{\alpha}$. 

Thus if an algorithm has approximation factor of $\alpha$, it must select $x_m$ with probability at least $0.9$ for all $1 \leq m \leq n$, since otherwise it incurs cost at least $d_m^2 > \alpha \c_k(\{x_1, x_2, \dots, x_m\}).$ While it is possible that centers chosen in the future may incur a smaller cost for $x_m$, because $n$ is unknown we can simply have the streaming stop at this point. The result follows.
\end{proof}

By combining Lemmas \ref{lem_diameter} and \ref{lem:seq_requirement}, we get a lower bound for the number of centers an online algorithm must select given a worst case ordering of a dataset $X$. 

\begin{thm}\label{thm:lower_bound}
Let $X$ be an arbitrary set of points. There exists an ordering of $X$ such that a streaming algorithm with  approximation factor $\alpha$ must select at least $0.9\lfloor\frac{\seq_k(X)}{k\lceil \log_2 \frac{1}{2}\sqrt{n\alpha}\rceil}\rfloor$ points in expectation.
\end{thm}

\section{Bounds on \texorpdfstring{$\seq_k(X)$}{Lg} for mixture distributions}

In this section, we consider the case where the input data $X$ is generated from some distribution $\D$ over $\R^d$. While each point $x_i \in X$ is independently sampled from $\D$, we make no assumptions about the order in which these points are presented to our algorithm. 
Furthermore, in the $k$-clustering setting, it is natural to assume that $\D$ is a mixture of $k$ distributions $\D^1, \D^2, \dots, \D^k$ over $\R^d$ such that each $\D^i$ corresponds to an ``intrinsic" cluster of $\D$. In particular, we will find bounds on $\seq_k(X)$ under the assumption that each $\D^i$ is a relatively well behaved distribution.

We begin by defining the \textit{aspect ratio} of a set $X$, which will subsequently be used to bound $\seq_k(X)$.

\begin{defn}
The aspect ratio of a set of points $S = \{s_1, s_2, \dots, s_n\}$, denoted $\asp(S)$, is the ratio between the distance of the farthest two points of $S$ and the closest two points of $S$. That is, $$\asp(S) = \frac{\max_{i \neq j}d(s_i, s_j)}{\min_{i \neq j}d(s_i, s_j)}.$$
\end{defn}

Let $X$ be a dataset drawn from $\D$. As shown in Lemma \ref{lem_diameter}, any $(2,k)$ subsequence have points with distances that grow exponentially. Furthermore, by the pigeonhole principle, at least $1/k$ of the elements in any $(2,k)$ sequence must come from some distribution $\D^i$. It follows that we can relate $\seq_k(X)$ to the aspect ratio $\asp(X^i)$, where $X^i$ denotes the points in $X$ drawn from $\D^i$. 

\begin{lem}\label{lem:asp_ratio}
For any set of points $X$ and any $k \geq 1$, if $\seq_k(X) \geq 1$, then for some $1 \leq i \leq k$, $\asp(X^i) \geq 2^{\seq_k(X)/k^2}.$
\end{lem}

\begin{proof}
Let $\{x_1, x_2, \dots, x_{\seq_k(X)/k}\} \subset X$ be the largest $(2, k)$ sequence in $X$ for which all $x_j$ are drawn from some $\D^i$. Such a sequence must exist by the definition of $\seq_k(X)$ and by a simple pigeonhole argument. 

For any $1 < m \leq \seq_k(X)$, let $d_m = \min_{1 \leq i \leq m-1}d(x_i, x_m)$. By the argument given in the proof of Lemma \ref{lem_diameter}, for any $m > k$  there exists $1 \leq i \leq k-1$ such that $d_m > 2d_{m-i}$. Thus applying this argument $\seq_k(X)/k$ times, we see that $d_{\seq_k(X)} > 2^{\seq_k(X)/k^2}d_{i}$ for some $2 \leq i \leq k$. The result follows from the definition of the aspect ratio.
\end{proof}

We will now show that for a broad class of distributions $\D$ over $\R^d$, namely those for which each $\D^i$ has finite variance and bounded probability density, that $\asp(X \sim \D^n) = O( n^3)$. This in turn will imply that $\seq_k(X) \leq O(k^2\log n)$. 

\begin{thm}\label{thm:asp_ratio_bound}
Let $\D$ be a distribution over $\R^d$ that is a mixture of distributions $\D^1, \D^2, \dots, \D^k$. Suppose there exist constants $D, \rho$ such that the following hold:
\begin{enumerate}
	\item For each $1 \leq i \leq k$, the expected squared distance from $x \sim \D^i$ to its mean is bounded. In particular, $$\E_{x \sim \D^i}[d(x, \E_{x' \sim \D}[x'])^2] \leq D,$$ for some $0 < D < \infty$.
	\item $\D$ (the entire mixture distribution) has probability density at most $\rho$ for some $0 < \rho < \infty$.
\end{enumerate}
Then for $X \sim \D^n$, with probability at least $1 - \frac{2}{n}$, $\seq_k(X) = O(k^2\log n)$. 
\end{thm}

Theorem \ref{thm:asp_ratio_bound} is proved through the following lemmas.

\begin{lem} \label{lem_VC}
(VC theory) Let $\X$ denote any probability distribution over $\R^d$. For any ball $B \subset \R^d$, let $\E[B]$ denote $\Pr_{x \sim \X}[x \in B]$. For any $0 < \delta < 1$, let $\alpha_n = \sqrt{\frac{4(d+2)\ln(16n/\delta)}{n}}$. Then with probability $1-\delta$ over $S \sim \X^n$, for \textbf{all} balls $B \subset \R^d$, $$\frac{|S \cap B|}{n} \leq \E[B] + \alpha_n^2 + \alpha_n\sqrt{\E[B]}.$$
\end{lem}

For a proof of Lemma \ref{lem_VC}, see Lemma 1 of \cite{dasgupta2007active}. 

\begin{lem}\label{lem:min_dist}
Let $\D$ be as described in Theorem \ref{thm:asp_ratio_bound}. Then for $X = \{x_1, x_2, \dots, x_n\} \sim \D^n$, with probability at least $1 - \frac{1}{n}$, $\min_{i \neq j} d(x_i, x_j) = \Omega(\frac{1}{n})$. 
\end{lem}

\begin{proof}
Let $\alpha_n = \sqrt{\frac{4(d+2)\ln(16n^2)}{n}}$ (we are setting $\delta = 1/n$ in the notation from Lemma \ref{lem_VC}). Let $r > 0$ be arbitrary. Then by Lemma \ref{lem_VC}, for all balls of radius $r$, we have that with probability $1-\frac{1}{n}$ over $X \sim \D^n$, $\frac{|X \cap B|}{n} \leq \Pr_{x \sim \D}[x \in B] + \alpha_n^2 + \alpha_n\sqrt{\Pr_{x \sim \D}[x \in B]}.$ We can bound $\Pr_{x \sim \D}[x \in B]$ by integrating the probability density of $\D$ over $B$. In particular, if $C_d$ denotes the volume of the unit ball in $\R^d$, we have that $\Pr_{x \sim \D}[x \in B] \leq \int_B \rho d\mu = \rho C_d r^d$ where $\rho$ is the upper bound on the probability density of $\D$. Substituting this, we see that with probability at least $1- \frac{1}{n}$, for all balls of radius $r$, $$\frac{|X \cap B|}{n} \leq \rho C_dr^d + \alpha_n^2 + \alpha_n\sqrt{\rho  C_dr^d}.$$ By setting $r = \Omega(\frac{1}{n})$, and observing that $\alpha_n^2 << 1/n$ as $n \to \infty$, we see that with probability at least $1 - \frac{1}{n}$, $|X \cap B| < 2$ for all balls of radius $r$. By the triangle inequality, this implies that $\min_{i \neq j} d(x_i, x_j) \geq 2r  = \Omega(\frac{1}{n})$, as desired. 
\end{proof}

\begin{lem}\label{lem:max_dist}
Let $\D$ be as described in Theorem \ref{thm:asp_ratio_bound}, and $X = \{x_1, x_2, \dots, x_n\} \sim \D^n$. Let $X^i$ be set of points in $X$ sampled from $\D^i$ (as described earlier). Then with probability at least $1 - \frac{1}{n}$, for all $1 \leq i \leq k$, $\max_{x_a\neq x_b \in X^i} d(x_a, x_b) \leq O(kn^2)$. 
\end{lem}

\begin{proof}
Fix any $1 \leq i \leq k$. By the triangle inequality, $$\max_{x_a\neq x_b \in X^i} d(x_a, x_b) \leq 2\max_{x_a \in X^i}d(x_a, \mu^i),$$ where $\mu^i = \E_{x \sim \D^i}[x]$. Also, $|X^i| \leq |X| = n$. Therefore, by markov's inequality, we see that $\Pr_{x \sim \D^i}[d(x, \mu)^2 > Dkn^2] \leq \frac{1}{kn^2}$. Therefore, by a union bound over all $x^a \in X^i$, with probability at least $1 - \frac{1}{nk}$, $d(x_a, \mu^i)^2 \leq Dkn^2$ for all $x^a \in X^i$. Taking a union bound over all $1 \leq i \leq k$, gives the desired result.
\end{proof}

We are now in the configuration to prove Theorem \ref{thm:asp_ratio_bound}.

\begin{proof}
(Theorem \ref{thm:asp_ratio_bound}) By Lemmas \ref{lem:min_dist} and \ref{lem:max_dist}, we have that for all $1 \leq i \leq k$, with probability at least $1 - \frac{2}{n}$, $\asp(X^i) \leq O(kn^3)$. By Lemma \ref{lem:asp_ratio}, we have that for some $i$, $\seq_k(X^i) \leq k^2\log(\asp(X))$ which implies that $\seq_k(X) \leq O(k^2\log n)$ as desired.
\end{proof}

As an immediate consequence of Theorem \ref{thm:asp_ratio_bound}, we see that for mixtures of $k$ Gaussians, as well as for mixtures of $k$ uniform distributions, $\seq_k(X)$ is $O(k^2\log n)$. 

\section{Conclusion and open questions}
We design a new $k$-means clustering algorithm in the online no-substitution setting, where importantly, points are received in arbitrary order. We introduce a new complexity measure, $\seq_k(X)$, to bound the number of centers the algorithm returns. We show that the complexity of data generated from many mixture distributions is bounded by $\seq_k(X)=O(k^2\log n)$. We prove that the algorithm takes only $O(\seq_k(X)\log(n)k\log(k))$ centers, and the algorithm is a $\poly(k)$-approximation. We complement this result by proving a lower bound of $\Omega\left(\frac{\seq_k(X)}{k\log(\alpha n)}\right)$ on the number of centers taken by any $\alpha$-approximation algorithm.

An obvious direction for future work is to improve the algorithm's parameters or prove they are tight. We proved that our algorithm is $\poly(k)$-approximation. Can it be improved to $\Theta(1)$-approximation?
For constant $k$ we bounded the number of centers, taken by our algorithm, by $O(\seq_k(X)\log(n))$ and showed a lower bound of $\Omega(\nicefrac{\seq_k(X)}{\log (n)})$, for any $\Theta(1)$-approximation algorithm. There is a gap of $\polylog(n)$ between our lower and upper bounds on the number of centers. An interesting future work would be to close this gap.

%Another natural future direction is to improve the memory size of Algorithm~\ref{alg:1}. Currently, it saves all examples observed. We conjecture the existence of an algorithm with similar parameters (number of centers and approximation), which saves only $\Theta(k)$ points in memory. 

The new algorithm and complexity measure are suggested to handle the case the order of the data is arbitrary, but can they also help if the order is random? It is known, \cite{moshkovitz2019unexpected}, that if the order is random, then $\Omega(\log n)$ centers are necessary. This lower bound was proved using a high $\mes$ complexity dataset, which is equal to $n$. Suppose the data's complexity is $\seq_k(X) \ll n$. Can the number of centers taken by a $\poly(k)$-approximation algorithm be dependent solely on $\seq_k(X)$ and not on $\log(n)$?

\section*{Acknowledgments} 

We thank Kamalika Chaudhuri for posing the central question of this paper: whether it is possible to obtain better results for online $k$-means clustering with adversarial order under assumptions on the underlying dataset (i.e. drawn from a mixture of Gaussians). We also thank Sanjoy Dasgupta for several helpful discussions about our results and proofs. 

Finally we thank NSF under CNS 1804829  for research support.

\bibliography{references}

\end{document}